\documentclass[12pt]{amsart}
\usepackage[margin=1.2in]{geometry}
\usepackage[sorting = none, backend = bibtex, style=numeric-comp]{biblatex}
\title[Identification and Recovery from Partial Trajectories]{Identification and Recovery of Linear Dynamics from Partial Trajectories}

\author[Le Gong]{Le Gong}

\date{\vspace{-5ex}}
\usepackage{amsthm, amsfonts, mathrsfs}
\usepackage{ dsfont }
\usepackage{amssymb}
\usepackage{enumitem}
\usepackage{graphicx}
\usepackage{subfig}

\usepackage{float}
\usepackage{bbm}
\usepackage{amsmath}
\usepackage{comment}
\usepackage{listings}

\usepackage{diagbox}
\usepackage[dvipsnames]{xcolor}
\usepackage{algorithm}
\usepackage{booktabs}
\usepackage{algorithmicx}
\usepackage[noend]{algpseudocode}
\usepackage{tikz}
\allowdisplaybreaks

\usepackage{hyperref}
\hypersetup{
    colorlinks,
    citecolor=blue,
    filecolor=blue,
    linkcolor=blue,
    urlcolor=blue
}
\numberwithin{equation}{section}
\newtheorem{theorem}[equation]{Theorem}
\newtheorem{proposition}[equation]{Proposition}
\newtheorem{lemma}[equation]{Lemma}

\newtheorem{corollary}[equation]{Corollary}
\newtheorem{definition}[equation]{Definition}
\newtheorem{remark}[equation]{Remark}

\newtheorem{assume}{Assumption}

\newcommand{\R}{\mathbb{R}}

\newcommand{\rank}{{\rm rank\,}}

\usepackage[foot]{amsaddr}

\definecolor{aacolor}{rgb}{0.05, 0.75, 1}

\newcommand{\Ocal}{\mathcal O}
\newcommand{\range}{\operatorname{range}}

\newcommand{\set}[1]{\left\{#1\right\}}
\newcommand{\paren}[1]{\left(#1\right)}
\newcommand{\Omegaall}{\boldsymbol\Omega}
\newcommand{\eps}{\varepsilon}
\newcommand{\pinv}{\dagger}

\newcommand{\sobs}{s_{\Omega}}

\begin{document}
\date{}
\address{\textrm{(Le Gong)}
	School of Mathematical Sciences,
Beijing University of Posts and Telecommunications, Beijing 100876, China}
\email{gongle@bupt.edu.cn}
\thanks{Corresponding author: Le Gong}

\begin{abstract}
We study the joint recovery of an unknown system matrix
$A\in\mathbb R^{n\times n}$ and a rank-$r$ initial-state matrix
$X_0\in\mathbb R^{n\times m}$ from partial observations of the state
matrices $A^tX_0$, for $t=0,\ldots,T$. Each trajectory is observed at a fixed set of state coordinates, which may differ across trajectories. We
derive necessary conditions for identifiability, including obstructions
arising from incomplete spatial coverage and insufficient dynamical
span, as well as lower bounds on the number of observations. We then
characterize the Jacobian of the measurement map, account for the
intrinsic symmetry of the low-rank factorization, and establish a
rank-certificate principle yielding generic local identifiability. When
the trajectories whose initial states form a basis of
$\operatorname{range}(X_0)$ are fully observed, we obtain necessary and
sufficient rank conditions for global recovery, together with explicit
reconstruction formulas. We formulate the joint recovery problem as a
nonlinear least-squares problem over the system matrix $A$ and the low-rank factors of $X_0$, and solve it using Adam. Numerical
experiments on synthetic and real-world data evaluate the recovery
performance under different spatial sampling rates and observation
horizons.
\end{abstract}

\maketitle
\section{Introduction}

Learning dynamical systems from incomplete observations is a fundamental problem in data-driven modeling. It arises in applications such as fluid dynamics, environmental transport, network diffusion, and epidemic propagation, where the underlying state is often high-dimensional, but measurements are available only at a limited number of spatial locations due to limited sensing resources, communication constraints, or restricted physical access. In such settings, the data are often temporally rich but spatially
incomplete. A central question is therefore whether temporal evolution can compensate for limited spatial sensing, enabling the recovery of both the underlying states and the governing dynamics. This question motivates the development of methodologies that exploit temporal information to overcome the limitations of sparse spatial observations.

Inspired by the spatial sampling framework in~\cite{lu2009spatial},
dynamical sampling was introduced in~\cite{aldroubi2013dynamical} and has
since emerged as a mathematical framework for studying the recovery of
dynamical systems from spatiotemporal measurements. In the classical setting, a signal
$x_t$ evolves in a Hilbert space $\mathcal{H}$ according to
\begin{equation}
\label{eq:intro-classical-dynamics}
x_{t+1}=Ax_t,
\end{equation}
where the evolution operator $A$ is assumed to be known. Given measurements
of the form
\[
Sx_t=SA^tx_0,
\qquad
t=0,\ldots,T,
\]
where $S$ denotes a sampling operator, the objective is to recover the unknown initial state $x_0$. The temporal evolution can propagate information from unobserved coordinates to observed ones, allowing a sequence of spatially sparse measurements to contain more information than a single static observation. In this way, dynamical sampling replaces dense spatial sensing at one time instant with a joint space--time sampling strategy.

In many applications, however, the governing dynamics are themselves unknown and must be inferred from data. This leads naturally to the problem of system identification. System identification aims to construct dynamical models from observed
input--output or state data~\cite{ljung1999system}. Many standard formulations of linear system identification estimate $A$ from fully observed state data, for example, using state-transition pairs $(x_t,x_{t+1})$. When the states are only partially observed, however, such state pairs are no longer directly available, and the system matrix cannot be estimated independently of the unobserved states. The evolution operator and the underlying trajectories must therefore be recovered jointly. The resulting problem is a coupled inverse problem that connects system identification with dynamical sampling, raising additional challenges in identifiability and reconstruction.
In this work, we consider multiple trajectories governed by a common
unknown matrix $A\in\mathbb R^{n\times n}$. Let
\[
X_0
=
\begin{bmatrix}
x_0^{(1)} & x_0^{(2)} & \cdots & x_0^{(m)}
\end{bmatrix}
\in\mathbb R^{n\times m}
\]
be the unknown matrix of initial states. The corresponding trajectories are
\[
x_t^{(i)}=A^tx_0^{(i)},
\qquad
i=1,\ldots,m,\quad t=0,\ldots,T.
\]
For each trajectory, we observe only a subset of its spatial coordinates:
\begin{equation}
\label{eq:intro-measurement-model}
y_t^{(i)}
=
S_i x_t^{(i)}
=
S_iA^tx_0^{(i)},
\qquad
i=1,\ldots,m,\quad t=0,\ldots,T,
\end{equation}
where $S_i$ is a row-selection sampling operator. The sampling operator may
vary across trajectories but is fixed over time for each individual
trajectory.

We assume that the initial-state matrix is low-rank:
\begin{equation}
\label{eq:intro-low-rank}
\operatorname{rank}(X_0)=r\ll\min\{n,m\},
\qquad
X_0=UV^\top,
\end{equation}
where $U\in\mathbb R^{n\times r}$ and
$V\in\mathbb R^{m\times r}$. This assumption models the situation in which the initial states belong to a common low-dimensional subspace. This low-rank structure couples the
trajectories, allowing information to be shared across different trajectories under sparse spatial observations.

Given the measurement model~\eqref{eq:intro-measurement-model}, our
central question is under what conditions the observations uniquely determine the pair $(A,X_0)$ and how this pair can be recovered in practice. We address this question from both theoretical and computational perspectives by analyzing local identifiability, establishing global
recovery under fully observed basis trajectories, and developing a
numerical method for recovery from partial trajectory data.

\subsection{Main contributions}
Our contributions are summarized as follows.
\begin{enumerate}[leftmargin=2em]
\item We identify spatial and dynamical obstructions to joint recovery
and derive measurement
requirements.
\item We characterize the measurement Jacobian, identify the
$r^2$-dimensional kernel induced by the nonuniqueness of the low-rank
factorization, and show that, for a fixed sampling configuration, a
single rank certificate yields generic local identifiability.
\item When the trajectories initialized from a basis of
$\operatorname{range}(X_0)$ are fully observed, we obtain necessary and
sufficient rank conditions for global recovery, together with explicit
reconstruction formulas for $A$ and $X_0$. We also derive generic
conditions on the observation horizon and the number of coordinates
observed for each remaining trajectory.
\item We formulate joint recovery as a nonlinear least-squares problem
over the system matrix $A$ and the low-rank factors of $X_0$, and solve it by direct joint optimization using Adam. Numerical
experiments on synthetic and real-world data investigate the effects of
the spatial sampling rate and observation horizon and compare the
proposed approach with independent factorized low-rank matrix
completion.
\end{enumerate}

\subsection{Related work}
\subsubsection{Dynamical sampling with known evolution operators} 
Classical dynamical sampling studies the recovery of an unknown initial
state $x_0$ from spatially incomplete measurements of its temporal evolution
under a known operator $A$
\cite{aldroubi2013dynamical,
aldroubi2017dynamical,
aldroubi2015exact,
aguilera2025optimal,
aldroubi2019continuous,
diazmartin2023error,
zhang2017dynamical}.
The literature primarily focuses on deterministic sampling
schemes and investigates the existence and stability of recovery, as
well as characterizations of admissible sampling configurations.
Randomized space--time sampling has subsequently been studied for
graph-based dynamical systems
\cite{huang2024robust,huang2025random}.

Dynamical sampling has also been extended to source-driven systems of
the form
\[
x_{t+1}=Ax_t+w_t,
\]
where $w_t$ represents an unknown input or forcing term. In this extended setting, the objective is to recover the source term or jointly recover the source and the initial state, while the
evolution operator is generally assumed to be known \cite{aldroubi2025periodic,
aldroubi2024characterization,
aldroubi2024spatially,
ruchi2026frames, gong2026randomized}. Related inverse-source problems for physical fields governed by partial differential equations have also been studied using spatiotemporal sensor
measurements \cite{murraybruce2015estimating,
murraybruce2017sampling}.

Our setting differs from these dynamical sampling problems in two key
respects. First, the evolution operator $A$ is unknown and must be learned
from the observations. Second, we consider multiple source-free trajectories whose initial states share a low-rank structure. Our goal is to jointly
recover the evolution operator $A$ and the initial-state matrix $X_0$ from partial spatiotemporal observations.

\subsubsection{System identification and operator learning}
In the fully observed linear setting, the system matrix can be estimated from consecutive state snapshots using least squares, a classical approach in system identification~\cite{ljung1999system}. Related data-driven methods, such as dynamic mode decomposition, approximate a linear evolution operator from snapshot pairs~\cite{tu2014dynamic}. Operator inference adopts a similar regression-based perspective to infer reduced-order operators from state trajectory data~\cite{peherstorfer2016operator}.

The recovery of an unknown evolution operator has also been studied within the framework of dynamical sampling. Tang~\cite{tang2017system} considered the simultaneous recovery of an unknown convolution operator and an unknown initial state from spatially subsampled observations. Under structural assumptions on the convolution filter and the sampling pattern, both quantities can be recovered using Prony-type methods.  

More recently, \cite{kuemmerle2024learning} studied the recovery of a transition
operator from sparse observations of the entries of its matrix powers
\[
A,A^2,\ldots,A^T.
\]
They used a block-Hankel embedding to reformulate the nonlinear recovery
problem as a structured low-rank matrix-completion problem and established
recovery guarantees under random space--time sampling.

Our setting differs from these works in several important respects. In contrast to the setting considered in~\cite{tang2017system}, where the evolution operator has a highly structured convolutional form, we allow $A$ to be a general matrix without imposing such structural assumptions. Moreover, rather than directly observing entries of the matrix powers $A^t$ as in \cite{kuemmerle2024learning}, we observe partial measurements of their action on unknown initial states,
\[
y_t^{(i)}=S_iA^tx_0^{(i)},
\]
where both the system matrix $A$ and the initial states $x_0^{(i)}$ are
unknown. We consider multiple trajectories whose initial states form a
low-rank matrix and study the joint recovery of $A$ and $X_0$ from
trajectory-dependent partial observations. In addition, each trajectory uses a fixed spatial sampling set throughout the observation interval, whereas the random space--time model in~\cite{kuemmerle2024learning} allows the sampled entries to vary with time.

\subsection{Paper organization}

Section~\ref{setting} introduces the observation model and derives
necessary conditions for identifiability. Section~\ref{main-result}
characterizes the measurement Jacobian and establishes a rank-based criterion for generic local identifiability. Section~\ref{sec:global}
develops exact global recovery results when the trajectories initialized
from a basis of $\operatorname{range}(X_0)$ are fully observed.
Section~\ref{sec:numerical-experiments} presents the direct joint
least-squares method based on Adam and reports numerical experiments on
synthetic and real-world data. The proofs of the main theoretical
results are collected in the appendices.

\section{Problem Setting}
\label{setting}
Fix an integer $T\geq1$ and let $A \in \mathbb{R}^{n\times n}$ be an unknown linear operator and let
\[
X_0 = \bigl[x_0^{(1)}, x_0^{(2)}, \dots, x_0^{(m)}\bigr] \in \mathbb{R}^{n\times m}
\]
denote an unknown matrix of initial states, where each column $x_0^{(i)} \in \mathbb{R}^n$ generates one trajectory. We assume that $X_0$ is low-rank, i.e.,
\[
\operatorname{rank}(X_0) = r,\quad 1\leq r\leq\min\{n,m\}, \qquad X_0 = U V^\top,
\]
with full-column-rank factors $U \in \mathbb{R}^{n\times r}$ and $V \in \mathbb{R}^{m\times r}$. The low-rank assumption is motivated by both practical and theoretical considerations. In many applications, high-dimensional spatiotemporal systems often exhibit low-dimensional dominant behaviors, meaning multiple observed trajectories naturally arise from a shared, lower-dimensional latent subspace. From a sampling perspective, the low-rank structure introduces redundancy across trajectories. If each initial state were treated independently, recovering the collection $\{x_0^{(i)}\}_{i=1}^m$ would generally require a large number of spatial observations for every trajectory. By exploiting the shared low-dimensional structure encoded in $X_0$, information obtained from one trajectory can assist in recovering others, potentially reducing the overall sampling burden.

The trajectories evolve according to the linear dynamics
\[
x_t^{(i)} = A^t x_0^{(i)}, \qquad t = 0,1,\dots,T,\quad i = 1,\dots,m.
\]
For each trajectory $i$, we observe the state through a row-selection sampling operator
\[
S_i \in \mathbb{R}^{m_{s,i}\times n}, \qquad i = 1,\dots,m,
\]
where each row of $S_i$ is a canonical basis vector in $\mathbb{R}^n$. Equivalently, there exists an index set $\Omega_i \subset \{1,\dots,n\}$ such that $S_i$ extracts the coordinates of $x \in \mathbb{R}^n$ indexed by $\Omega_i$, i.e.,
\[
S_i x = x|_{\Omega_i}.
\]
The set $\Omega_i$ is fixed over time, and $m_{s,i}=|\Omega_i|$.
Thus the measurements are given by
\[
y_t^{(i)} = S_i x_t^{(i)} = S_i A^t x_0^{(i)}, \qquad t = 0,1,\dots,T.
\]

Given the sampling operators $\{S_i\}_{i=1}^m$ and the noiseless data
$y_t^{(i)}$ for $i=1,\ldots,m$ and $t=0,\ldots,T$, we aim to reconstruct
the unknown pair $(A,X_0)$.

\subsection{Observation map and intrinsic parameter space}
For a positive integer $d$, write
\[
[d]:=\{1,\ldots,d\}.
\]
We encode the spatial sampling pattern by
\[
\Omegaall
:=
\{(k,i)\in[n]\times[m]:k\in\Omega_i\}.
\]
Thus, $(k,i)\in\Omegaall$ means that the $k$th coordinate of the
$i$th trajectory is observed. The number of observed
coordinate--trajectory pairs and the total number of measurements are, respectively,
\[
\sobs
:=
|\Omegaall|
=
\sum_{i=1}^m m_{s,i},
\qquad
p
:=
(T+1)\sobs.
\]

For $X\in\mathbb R^{n\times m}$, let
\[
P_{\Omegaall}(X)\in\mathbb R^{\sobs}
\]
denote the vector consisting of the entries $X_{ki}$ with
$(k,i)\in\Omegaall$, ordered first by trajectory index and then
according to the row order of $S_i$.

Define the measurement map
\begin{equation}\label{eq:intrinsic-map}
 \Psi(A,X)=\big(P_{\Omegaall}(A^tX)\big)_{t=0}^T\in\mathbb R^p,
 \qquad (A,X)\in\mathbb R^{n\times n}\times\mathcal M_r,
\end{equation}
where $\mathcal M_r=\{X\in\mathbb R^{n\times m}:\operatorname{rank}X=r\}$.
For $X=UV^\top\in\mathcal M_r$, the factors
$U\in\mathbb R^{n\times r}$ and $V\in\mathbb R^{m\times r}$ have full
column rank. However, the factorization is not unique. Indeed, for every invertible matrix
$Q\in\mathrm{GL}(r)$,
\[
(UQ)(VQ^{-\top})^\top=UV^\top.
\]
Thus, although $(U,V)$ contains $nr+mr$ parameters, this
factorization freedom accounts for $r^2$ redundant degrees of freedom. Consequently,
\begin{equation}
\label{eq:rank-manifold-dim}
\dim\mathcal M_r
=
nr+mr-r^2
=
r(n+m-r).
\end{equation}
The tangent space of $\mathcal M_r$ at $X=UV^\top$ is
\begin{equation}
\label{eq:tangent-rank-manifold}
T_X\mathcal M_r
=
\left\{
\Delta U V^\top+U\Delta V^\top:
\Delta U\in\mathbb R^{n\times r},\
\Delta V\in\mathbb R^{m\times r}
\right\}.
\end{equation}
Let
\[
D\Psi(A,X_0):
\mathbb R^{n\times n}\times T_{X_0}\mathcal M_r
\longrightarrow \mathbb R^p
\]
denote the differential of the measurement map at $(A,X_0)$.

\begin{definition}\label{def:identifiability}
A pair $(A,X_0)$ is globally identifiable if
$\Psi(\widetilde A,\widetilde X_0)=\Psi(A,X_0)$ and
$\operatorname{rank}\widetilde X_0=r$ imply
$(\widetilde A,\widetilde X_0)=(A,X_0)$. It is locally identifiable
if this implication holds in a neighborhood of $(A,X_0)$, and
infinitesimally identifiable if $D\Psi(A,X_0)$ is injective on
$\mathbb R^{n\times n}\times T_{X_0}\mathcal M_r$.
\end{definition}

We will also use the following notion of genericity. A property holds generically if its failure is contained in the zero
set of a polynomial that is not identically zero in the specified
parameter coordinates. This exceptional set has Lebesgue measure zero in those coordinates.

\subsection{Spatial coverage and dynamical span}
A fundamental question is whether the pair $(A,X_0)$ is uniquely
determined by the observations. We first identify two obstructions to
identifiability, arising from incomplete spatial coverage and
insufficient dynamical span, and then formulate the corresponding
necessary assumptions for the subsequent analysis.

\begin{proposition}
\label{prop:nonuniqueness-fixed-S}
Let
$
\mathcal S
:=
\begin{bmatrix}
S_1\\
\vdots\\
S_m
\end{bmatrix}
\in\mathbb R^{s_\Omega\times n}$
be the aggregated sampling operator and define
\[
 G_{\mathcal S}:=\{R\in\mathrm{GL}(n):S_iR=S_i
 \text{ for every }i\in[m]\}.
\]
For every $R\in G_{\mathcal S}$, the pair
\begin{equation}
\label{eq:similarity-ambiguity}
  \widetilde A=RAR^{-1},
  \qquad
  \widetilde X_0=RX_0
\end{equation}
produces the same observations as $(A,X_0)$ at every time. If
$\ker\mathcal S\neq\{0\}$, then $G_{\mathcal S}$ contains nonidentity
matrices arbitrarily close to $I_n$.

\end{proposition}

For row-selection operators, $\ker\mathcal S=\{0\}$ is equivalent to
$\bigcup_i\Omega_i=[n]$. If this coverage condition fails,
Proposition~\ref{prop:nonuniqueness-fixed-S} shows that every pair
$(A,X_0)$ has arbitrarily close, distinct pairs producing exactly the
same observations at all times. Consequently, $(A,X_0)$ is neither locally nor globally identifiable. Full spatial coverage is therefore necessary for identifiability, which motivates the following assumption.

\begin{assume}
\label{ass:full-coverage}
The sampling sets cover all coordinates:
\[
\bigcup_{i=1}^m\Omega_i=[n].
\]
\end{assume}

\begin{remark}
When $m=1$, Assumption~\ref{ass:full-coverage} reduces to
\[
\Omega_1=[n].
\]
Thus, the single trajectory must be fully observed at every time.
Temporal observations from a fixed proper subset of coordinates cannot
compensate for incomplete spatial coverage when both $A$ and the initial state are unknown.
\end{remark}

Assumption~\ref{ass:full-coverage} is necessary but does not by itself
guarantee identifiability. Even with complete spatial observation, the
states used in the transitions must span the entire state space.
Define
\[
 \mathcal K_{T-1}(A,X_0)
 =\operatorname{span}\{A^kx_0^{(i)}:i\in[m],\ k=0,\ldots,T-1\}.
\]
The corresponding matrix is
\begin{equation}\label{eq:krylov-minus-general}
 K_-(A,X_0)=[X_0,AX_0,\ldots,A^{T-1}X_0]\in\mathbb R^{n\times mT}.
\end{equation}

\begin{proposition} 
\label{prop:krylov-obstruction}
Let $s=\rank K_-(A,X_0)$.  If $s<n$, then there exists a nonzero matrix $E$ satisfying
\[
EK_-(A,X_0)=0.
\]
Consequently, $\widetilde A=A+E\neq A$ generates the same trajectories
through time $T$:
\[
\widetilde A^tX_0=A^tX_0,
\qquad t=0,\ldots,T.
\]
Thus, $A$ cannot be uniquely determined, even if every trajectory is
fully observed.
\end{proposition}

The following assumption records this necessary condition.
\begin{assume}
\label{ass:full-dynamical-span}
The states generated by the trajectories at times $0,\ldots,T-1$ span
the entire state space:
\[
\mathcal K_{T-1}(A,X_0)=\mathbb R^n,
\qquad\text{equivalently,}\qquad
\rank K_-(A,X_0)=n.
\]
\end{assume}

Since $K_-(A,X_0)$ contains $T$ blocks, each of rank at most $r$,
$
\operatorname{rank}K_-(A,X_0)\leq rT.
$
Therefore, Assumption~\ref{ass:full-dynamical-span} requires
\begin{equation}
\label{eq:rT-necessary}
rT\geq n.
\end{equation}
We likewise impose this assumption throughout the remainder of the
paper.

\subsection{Necessary measurement inequalities}
Let $B_0\in\mathbb R^{n\times r}$ be a column submatrix of $X_0$
whose columns form a basis for $\operatorname{range}(X_0)$. For each
$i\in[m]$, define the restricted observability matrix
\begin{equation}\label{eq:trajectory-observability}
 \mathcal O_{i,T}(A,B_0)=
 \begin{bmatrix}S_iB_0\\S_iAB_0\\\vdots\\S_iA^TB_0\end{bmatrix}
 \in\mathbb R^{(T+1)m_{s,i}\times r}.
\end{equation}
Note that its rank is independent of the choice of basis.

The following theorem collects several necessary conditions for
injectivity of the differential of the measurement map $\Psi$. In particular, if any one of these conditions fails, then $D\Psi(A,X_0)$ has a nontrivial kernel.

\begin{theorem}
\label{thm:necessary-laws}
Suppose $D\Psi(A,X_0)$ is injective on
$\R^{n\times n}\times T_{X_0}\mathcal M_r$.  Then
\begin{align}
&(T+1)\sobs\geq n^2+r(n+m-r),
\qquad
T\sobs\geq n^2,
\label{eq:dimension-lower}\\
&\operatorname{rank}\mathcal O_{i,T}(A,B_0)=r
\qquad\text{for every }i\in[m].
\label{eq:necessary-O}
\end{align}
Consequently,
\begin{equation}
\label{eq:combined-Q-lower}
  \sobs\geq
  \max\set{
  n,
  \left\lceil\frac{n^2}{T}\right\rceil,
  \left\lceil\frac{n^2+r(n+m-r)}{T+1}\right\rceil,
  m\left\lceil\frac{r}{T+1}\right\rceil
  }.
\end{equation}
\end{theorem}


\section{Jacobian characterization and generic local identifiability}
\label{main-result}

In this section, we derive the Jacobian of the factorized measurement
map and characterize the kernel induced by the nonuniqueness of the
factorization $X_0=UV^\top$. We then establish a generic rank principle
that provides a criterion for local identifiability of the pair
$(A,X_0)$.

\subsection{Perturbation analysis}
For each trajectory $i\in\{1,\dots,m\}$ and time $t\in\{0,\dots,T\}$, define
\begin{equation}\label{eq:blockPhi}
\Phi_{i,t}(A,U,V):= S_i A^t U v_i\in\mathbb{R}^{m_{s,i}},\qquad v_i := V^\top e_i\in\mathbb{R}^r,
\end{equation}
where $S_i\in\mathbb{R}^{m_{s,i}\times n}$ is a row-selection operator, and $e_i\in\mathbb{R}^m$ is the $i$-th standard basis vector.
Stack the blocks in increasing time order, and within each time in increasing trajectory order, into one vector
\[
\Phi(A,U,V) := \big(\Phi_{i,t}(A,U,V)\big)_{i=1,\dots,m;\ t=0,\dots,T}\in\mathbb{R}^{p},
\qquad
p=(T+1)\sobs.
\]

Let $(\Delta A,\Delta U,\Delta V)$ be perturbations and set $\Delta v_i := (\Delta V)^\top e_i$.
Applying the product rule together with the linearity of the Fréchet derivative to \eqref{eq:blockPhi} yields
\begin{equation}\label{eq:differential_rule}
\delta y_{i,t}
:= D\Phi_{i,t}(A,U,V)[\Delta A,\Delta U,\Delta V]
=
S_i\Big(
\delta(A^t)\,U v_i
+
A^t(\Delta U)\,v_i
+
A^t U\,\Delta v_i
\Big).
\end{equation}

For $t\ge 1$, the Fr\'echet derivative of the matrix power map $A\mapsto A^t$ is given by
\begin{equation}\label{eq:DApower}
D(A^t)[\Delta A]
=
\sum_{k=0}^{t-1} A^{t-1-k}(\Delta A)A^{k},
\qquad
D(A^0)[\Delta A]=0.
\end{equation}
Substituting \eqref{eq:DApower} into \eqref{eq:differential_rule} yields the following formula for the Jacobian operator:
\begin{equation}\label{eq:Jacobian_operator_block}
\begin{aligned}
&D\Phi_{i,t}(A,U,V)[\Delta A,\Delta U,\Delta V]\\
&\quad=S_i\left(
\sum_{k=0}^{t-1} A^{t-1-k}(\Delta A)A^kU v_i
+A^t(\Delta U)v_i+A^tU(\Delta V)^\top e_i\right).
\end{aligned}
\end{equation}
This defines the Jacobian as a linear operator mapping $(\Delta A,\Delta U,\Delta V)$ to
$\delta y_{i,t}$, without forming any large matrices explicitly.

Furthermore, to obtain an explicit Jacobian matrix $J(A,U,V):=D\Phi(A,U,V)$ satisfying
\[
\delta\Phi = J(A,U,V)
\begin{bmatrix}
\mathrm{vec}(\Delta A)\\[2pt]
\mathrm{vec}(\Delta U)\\[2pt]
\mathrm{vec}(\Delta V)
\end{bmatrix},
\]
we use the identity $\mathrm{vec}(M(\Delta A)b)=(b^\top\otimes M)\mathrm{vec}(\Delta A)$ where $\mathrm{vec}(\cdot)$ denotes column-wise vectorization and $\otimes$ denotes the Kronecker product.
Define
\[
b_{i,k}:=A^kU v_i\in\mathbb{R}^n,
\qquad
M_{i,t,k}:=S_iA^{t-1-k}\in\mathbb{R}^{m_{s,i}\times n}.
\]
Then the block row corresponding to $(i,t)$ decomposes as
\begin{equation}\label{eq:blockrow_decomp}
J^{(i,t)} = \big[\,J_A^{(i,t)}\ \ J_U^{(i,t)}\ \ J_V^{(i,t)}\,\big]
\in\mathbb{R}^{m_{s,i}\times (n^2+nr+mr)},
\end{equation}
where
\begin{align}
J_A^{(i,t)}
&=
\sum_{k=0}^{t-1}\left(b_{i,k}^\top\otimes M_{i,t,k}\right)
\in\mathbb{R}^{m_{s,i}\times n^2},
\label{eq:JA_block}
\\[4pt]
J_U^{(i,t)}
&=
\left(v_i^\top\otimes S_iA^t\right)
\in\mathbb{R}^{m_{s,i}\times (nr)},
\label{eq:JU_block}
\\[4pt]
J_V^{(i,t)}
&=
(S_iA^tU)\,(I_r\otimes e_i^\top)
\in\mathbb{R}^{m_{s,i}\times (mr)}.
\label{eq:JV_block}
\end{align}

The full Jacobian matrix is obtained by stacking the block rows $J^{(i,t)}$ for all $i=1,\dots,m$ and $t=0,\dots,T$ in the same order as $\Phi$, yielding
\[
J(A,U,V) = D\Phi(A,U,V)
=
\begin{bmatrix}
J^{(1,0)}\\ \vdots\\ J^{(m,0)}\\
J^{(1,1)}\\ \vdots\\ J^{(m,1)}\\
\vdots\\
J^{(1,T)}\\ \vdots\\ J^{(m,T)}
\end{bmatrix}.
\]
 
\subsection{Gauge kernel and injectivity}
Define the factorization map
\begin{equation}
\label{eq:pi-factorization}
\pi(A,U,V)=(A,UV^\top),
\end{equation}
so that $\pi(A,U,V)=(A,X_0)$ whenever $X_0=UV^\top$.
This map relates perturbations in the factor variables $(A,U,V)$ to
perturbations in the original variables $(A,X_0)$. The following
proposition characterizes its differential and relates injectivity of
$D\Psi$ to the kernel of $D\Phi$.

\begin{proposition}
\label{prop:intrinsic-vs-gauge}
Let $X_0=UV^\top$ with $U$ and $V$ full column rank. Then
\begin{equation}
\label{eq:Dpi}
  D\pi(A,U,V)[\Delta A,\Delta U,\Delta V]
  =\paren{\Delta A,\Delta UV^\top+U\Delta V^\top}.
\end{equation}
The image of $D\pi$ is
$\R^{n\times n}\times T_{UV^\top}\mathcal M_r$, and
\begin{equation}
\label{eq:ker-Dpi}
  \ker D\pi(A,U,V)
  =\set{(0,UH,-VH^\top):H\in\R^{r\times r}}.
\end{equation}
Furthermore, since $\Phi=\Psi\circ\pi$, the following are equivalent:
\begin{enumerate}[label=\textup{(\roman*)}]
\item $D\Psi(A,X_0)$ is injective on
$\R^{n\times n}\times T_{X_0}\mathcal M_r$;
\item $\ker D\Phi(A,U,V)
  =\set{(0,UH,-VH^\top):H\in\R^{r\times r}}.$
\end{enumerate}
\end{proposition}

\subsection{Generic rank principle and local identifiability}

Set $d=n^2+r(n+m-r)$ which is the dimension of $\mathbb R^{n\times n}\times\mathcal M_r$. By Proposition~\ref{prop:intrinsic-vs-gauge}, the $r^2$-dimensional
gauge subspace is always contained in the kernel of the Jacobian. Consequently,
\[
\operatorname{rank}D\Phi(A,U,V)
\leq
n^2+nr+mr-r^2=d.
\]
The following result shows that if this maximal rank is attained at one
full-rank factor point, then it is attained generically.

\begin{theorem}
\label{thm:generic-jacobian-rank}

Fix $n,m,r,T$ and the sampling operators $\{S_i\}_{i=1}^m$.
For $\theta=(A,U,V)$, let
\[
J(\theta):=D\Phi(A,U,V).
\]
Assume there exists a parameter point
$\theta_0=(A_0,U_0,V_0)$ such that
\begin{equation}
\label{eq:rank-certificate}
\begin{aligned}
&\operatorname{rank}(U_0)
=
\operatorname{rank}(V_0)
=
r,\\
&\operatorname{rank}J(\theta_0)=n^2+nr+mr-r^2.
\end{aligned}
\end{equation}
Then, for Lebesgue-almost-every \(  \theta=(A,U,V)\) satisfying $\operatorname{rank}(U) = \operatorname{rank}(V)= r,$ we have
\[ \ker J(\theta)  =   \left\{   (0,UH,-VH^\top)  :   H\in\mathbb R^{r\times r}  \right\}.\]
Equivalently, generically, \[ \operatorname{rank}J(\theta)= n^2+nr+mr-r^2, \qquad \dim\ker J(\theta)=r^2.\]
\end{theorem}

\begin{remark}
The theorem is conditional on the existence of a maximal-rank
certificate and does not characterize all sampling patterns for which
such a certificate exists. In Section~\ref{sec:basiss}, we verify this
condition for a special sampling configuration containing fully observed basis trajectories under explicit rank conditions. Characterizing more general sampling configurations is left for future work.
\end{remark}

\begin{corollary}
\label{cor:local_identifiability}
Let $\theta_\star=(A_\star,U_\star,V_\star)$ have full-column-rank factors and
suppose
\begin{equation}
\label{eq:pointwise-max-rank}
  \rank J(\theta_\star)=d.
\end{equation}
Then there exists a neighborhood $\mathcal W$ of $\theta_\star$ such that any
$\widetilde\theta=(\widetilde A,\widetilde U,\widetilde V)\in\mathcal W$
with
\[
  \Phi(\widetilde\theta)=\Phi(\theta_\star)
\]
satisfies
\[
  \widetilde A=A_\star,
  \qquad
  (\widetilde U,\widetilde V)
  =(U_\star Q,V_\star Q^{-\top})
\]
for a unique $Q\in GL(r)$ sufficiently close to $I_r$.
\end{corollary}

Consequently, the pair $(A,X_0)$ is locally identifiable. The only
remaining local nonuniqueness is the change of basis in the factors
$(U,V)$, which leaves $X_0=UV^\top$ unchanged.

\section{Global identifiability and explicit recovery}
\label{sec:global}
The Jacobian criterion established above is local. In this section, we
study global recovery in the special setting where $r$ trajectories are
fully observed and their initial states form a basis for
$\operatorname{range}(X_0)$. We derive necessary and sufficient rank
conditions for global identifiability, provide explicit reconstruction
formulas, and show that the necessary dimension conditions on the observation
horizon and spatial sampling are also generically sufficient.

\subsection{Exact recovery with fully observed basis trajectories}
\label{sec:basiss}

Let $\mathcal I\subseteq[m]$ be an index set of cardinality $r$, and
define
\begin{equation}
\label{eq:basis-B0}
B_0:=X_0[:,\mathcal I]\in\mathbb R^{n\times r}.
\end{equation}
We assume that $\operatorname{rank}(B_0)=r$, so the columns of $B_0$
form a basis for $\operatorname{range}(X_0)$. We refer to the
trajectories indexed by $\mathcal I$ as the basis trajectories and
suppose that they are fully observed, that is,
\[
S_j=I_n,
\qquad j\in\mathcal I.
\]

Since
$\range(X_0)=\range(B_0)$, each initial state has a unique coefficient
$c_i\in\R^r$ such that
\begin{equation}
\label{eq:basis-coeff}
  x_0^{(i)}=B_0c_i.
\end{equation}
For $t=0,\ldots,T$, define
\[
  B_t:=A^tB_0,
\]
and set
\begin{equation}
\label{eq:basis-Kpm}
  K_{B,-}:=\begin{bmatrix}B_0&B_1&\cdots&B_{T-1}\end{bmatrix},
  \qquad
  K_{B,+}:=\begin{bmatrix}B_1&B_2&\cdots&B_T\end{bmatrix}.
\end{equation}
For $i\notin \mathcal I$, define $\Ocal_{i,T}(A,B_0)$ by
\eqref{eq:trajectory-observability}.

Because the basis trajectories are fully observed, the matrices
$B_t$, $K_{B,-}$, and $K_{B,+}$ are directly determined by the
measurements. The following theorem gives necessary and sufficient rank
conditions under which $A$ and all remaining initial states can be
reconstructed uniquely from the observations. In the reconstruction formulas below, $M^\pinv$ denotes the Moore--Penrose pseudoinverse of a matrix $M$.

\begin{theorem}
\label{thm:basis-iff}
\label{ex:basis-sampling}
Let
\[
B_0:=X_0[:,\mathcal I]\in\mathbb R^{n\times r}.
\]
Assume that $\operatorname{rank}(B_0)=r$ and that the trajectories
indexed by $\mathcal I$ are fully observed. Then $(A,X_0)$ is globally
identifiable if and only if
\begin{equation}
\label{eq:basis-rank-conditions}
  \rank K_{B,-}=n,
  \qquad
  \rank\Ocal_{i,T}(A,B_0)=r
  \quad\text{for every }i\notin \mathcal I.
\end{equation}
When these conditions hold, reconstruction is explicit:
\begin{equation}
\label{eq:basis-A-recovery}
  A=K_{B,+}K_{B,-}^\pinv,
\end{equation}
and
\begin{equation}
\label{eq:basis-ci-recovery}
  c_i=
  \Ocal_{i,T}(A,B_0)^\pinv
  \begin{bmatrix}y_0^{(i)}\\ \vdots\\ y_T^{(i)}\end{bmatrix},
  \qquad
  x_0^{(i)}=B_0c_i
  \quad(i\notin \mathcal I).
\end{equation}
\end{theorem}

\begin{remark}
When the rank conditions in~\eqref{eq:basis-rank-conditions} hold,
$D\Psi(A,X_0)$ is injective. Consequently,
\[
\operatorname{rank}J(A,U,V)
=
n^2+r(n+m-r)
\]
for every full-column-rank factorization $X_0=UV^\top$.
\end{remark}

Note that in this setting, the minimum recovery horizon is the smallest $T\geq1$ for
which the rank conditions in~\eqref{eq:basis-rank-conditions} hold. If
no such finite $T$ exists, then $(A,X_0)$ cannot be globally recovered
over any finite observation horizon.

\subsection{Generic sufficiency of the dimension conditions}

The matrix $K_{B,-}$ has $rT$ columns, while
$\mathcal O_{i,T}(A,B_0)$ has $(T+1)m_{s,i}$ rows. Therefore, the rank
conditions in~\eqref{eq:basis-rank-conditions} require
\[
rT\geq n,
\qquad
(T+1)m_{s,i}\geq r
\quad(i\notin\mathcal I).
\]
The following result shows that these necessary dimension conditions
are also generically sufficient.

\begin{corollary}
\label{cor:generic-basis}
Fix an index set $\mathcal I\subseteq[m]$ of cardinality $r$, and
suppose that the trajectories indexed by $\mathcal I$ are fully
observed. If
\begin{equation}
\label{eq:basis-dimension-conditions}
rT\geq n,
\qquad
(T+1)m_{s,i}\geq r
\quad(i\notin\mathcal I),
\end{equation}
then the rank conditions in~\eqref{eq:basis-rank-conditions} hold
generically among pairs $(A,X_0)$ satisfying
\[
\operatorname{rank}(X_0)
=
\operatorname{rank}\bigl(X_0[:,\mathcal I]\bigr)
=
r.
\]
Consequently, such pairs are globally identifiable and satisfy the
maximal Jacobian-rank condition.
\end{corollary}

\section{Numerical experiments}
\label{sec:numerical-experiments}
The preceding theory concerns identifiability of the joint recovery problem. In this section, we investigate its computational behavior by directly minimizing the joint observation loss with Adam
\cite{kingma2015adam}. Synthetic experiments examine the effects of the spatial sampling rate and the observation horizon on the recovery of $A$ and $X_0$. A sea-surface-temperature experiment then evaluates the practical relevance of the model by reconstructing unobserved measurements from partial trajectories and
comparing the proposed method with independent low-rank matrix
completion. 


For a prescribed factorization rank $r$, we parameterize the
initial-state matrix as
\[
X_0=UV^\top,
\qquad
U\in\mathbb R^{n\times r},
\qquad
V\in\mathbb R^{m\times r}.
\]
Writing $v_i=V^\top e_i$, we minimize
\begin{equation}
\label{eq:numerical-objective}
\begin{aligned}
\mathcal L(A,U,V)
={}&
\frac{1}{2N_{\mathrm{obs}}}
\sum_{i=1}^m\sum_{t=0}^T
\left\|
S_iA^tUv_i-y_t^{(i)}
\right\|_2^2
+\frac{\lambda_A}{2}\|A\|_F^2  \\
&+
\frac{\lambda_{UV}}{2}
\left(
\|U\|_F^2+\|V\|_F^2
\right),
\end{aligned}
\end{equation}
where
\[
N_{\mathrm{obs}}
=
(T+1)\sum_{i=1}^m m_{s,i}.
\]
Here, small regularization terms are included for numerical stability.

For given $A$, $U$, and $V$, the predicted state matrices are generated recursively by
\[
X_0=UV^\top,
\qquad
X_{t+1}=AX_t,
\qquad
t=0,\ldots,T-1.
\]
Gradients of~\eqref{eq:numerical-objective} are computed through this
recursion, and $A$, $U$, and $V$ are updated jointly using Adam. The
complete procedure is summarized in
Algorithm~\ref{alg:direct-joint-adam}.

Since the objective is nonconvex, the optimization result may depend on the initialization. We therefore run the algorithm from three
independent random initializations in all experiments. For each run, we
retain the iterate with the smallest residual on the observed
measurements. The final estimate is the retained iterate with the
smallest residual among the three runs.

\begin{algorithm}[htbp]
\caption{Direct joint estimation by Adam}
\label{alg:direct-joint-adam}
\small
\begin{algorithmic}[1]
\Require Measurements
$\{(S_i,\{y_t^{(i)}\}_{t=0}^T)\}_{i=1}^m$,
rank $r$, iteration count $K$, and $N_{\mathrm{init}}$ initializations
\For{$s=1,\ldots,N_{\mathrm{init}}$}
    \State Randomly initialize $A$, $U$, and $V$
    \For{$k=1,\ldots,K$}
        \State Set $X_0=UV^\top$
        \For{$t=0,\ldots,T-1$}
            \State $X_{t+1}\gets AX_t$
        \EndFor
        \State Evaluate $\mathcal L(A,U,V)$ in
        \eqref{eq:numerical-objective}
        \State Update $A$, $U$, and $V$ jointly by Adam using
        backpropagation through the recursion
        \State Apply gradient clipping and the prescribed learning-rate
        schedule
        \State \textbf{if } $\rho(A)>\rho_{\max}$ \textbf{ then }
$A\gets \rho_{\max}A/\rho(A)$
        \State Retain the current iterate if its observed-data residual
        is the smallest obtained in this run
    \EndFor
\EndFor
\State Select the retained run having the smallest observed-data
residual
\State \Return $\widehat A$ and
$\widehat X_0=\widehat U\widehat V^\top$
\end{algorithmic}
\end{algorithm}

The spectral-radius rescaling in
Algorithm~\ref{alg:direct-joint-adam} is used only in the synthetic experiments. No such constraint is imposed on the sea-surface-temperature data. The remaining numerical parameters are specified separately in the corresponding experimental
subsections.

\subsection{Synthetic experiments}
\label{subsec:synthetic-experiments}

We generate a synthetic system with
\[
n=10,\qquad m=20,\qquad r=5.
\]
The entries of $A$, $U$, and $V$ are drawn independently from Gaussian
distributions, and $A$ is rescaled so that
\[
\rho(A)=0.8,
\]
where $\rho(A)$ denotes the spectral radius of $A$. This ground-truth pair
is then held fixed throughout all synthetic experiments. The noiseless state
matrices are then generated according to
\[
X_0=UV^\top,
\qquad
X_t=A^tX_0=A^tUV^\top,
\qquad
t=0,\ldots,T.
\]
During optimization, the current estimate of $A$ is rescaled after each
Adam update whenever its spectral radius exceeds $0.99$. Reconstruction
accuracy is measured by the relative errors
\begin{equation}
\label{eq:synthetic-errors}
\operatorname{RE}(A)
=
\frac{\|\widehat A-A\|_F}{\|A\|_F},
\qquad
\operatorname{RE}(X_0)
=
\frac{\|\widehat X_0-X_0\|_F}{\|X_0\|_F}.
\end{equation}

In the first experiment, we fix the observation horizon at $T=7$ and
vary the number of sampled coordinates per trajectory from $m_s=1$ to
$m_s=9$, corresponding to the sampling rates
\[
\frac{m_s}{n}=0.1,0.2,\ldots,0.9.
\]
For each sampling rate, every trajectory is observed on the same set of
$m_s$ coordinates at all observation times. A separate sampling pattern
is generated for each rate, without requiring the patterns at different
rates to be nested. Each pattern provides full spatial coverage across
the collection of trajectories. As shown in
Figure~\ref{fig:synthetic-sampling-rate}, the relative errors of both
$A$ and $X_0$ decrease substantially overall as the sampling rate
increases.

\begin{figure}[t]
\centering
\begin{minipage}{0.48\textwidth}
  \centering
  \includegraphics[width=\linewidth]{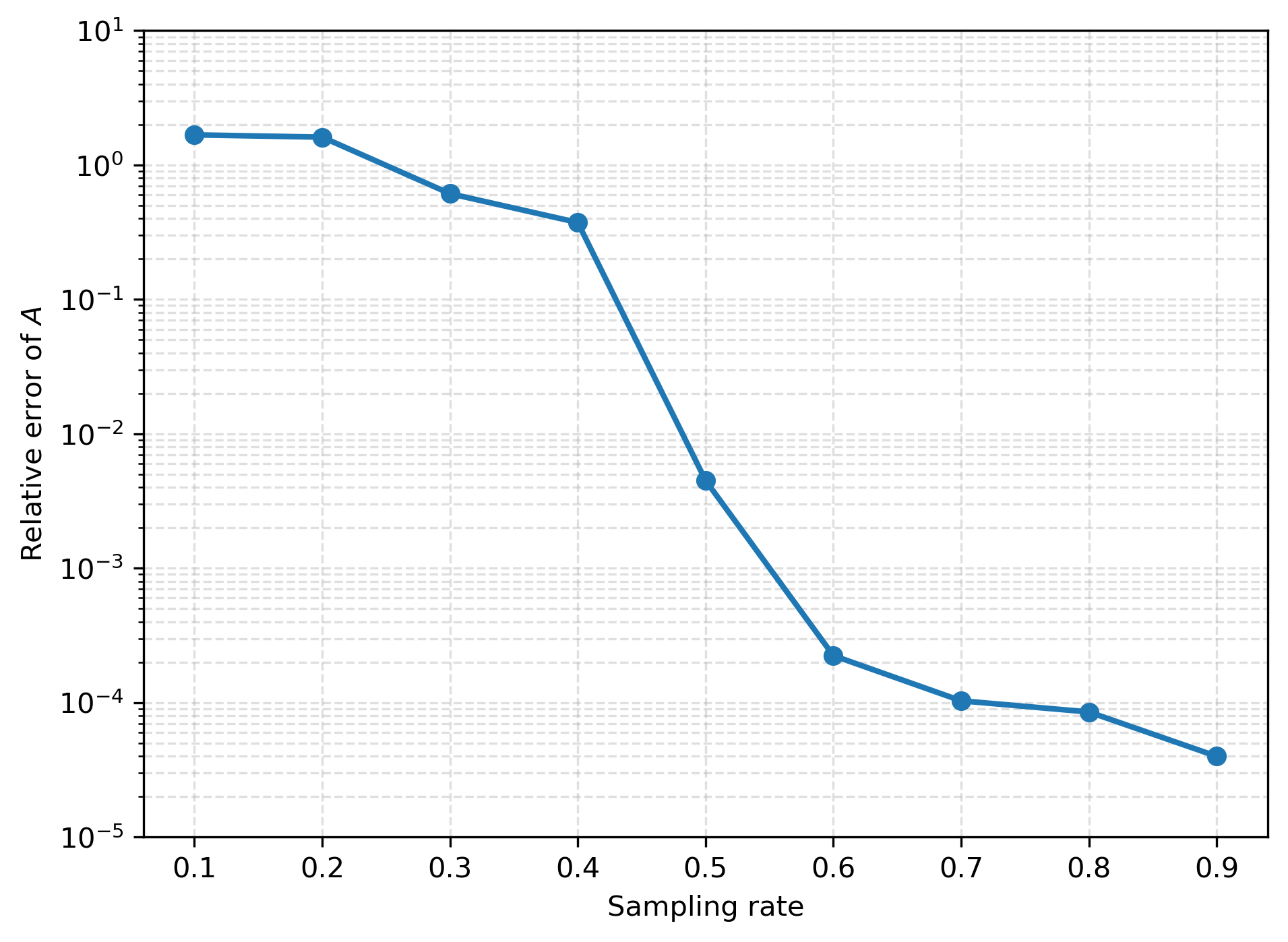}
\end{minipage}\hfill
\begin{minipage}{0.48\textwidth}
  \centering
  \includegraphics[width=\linewidth]{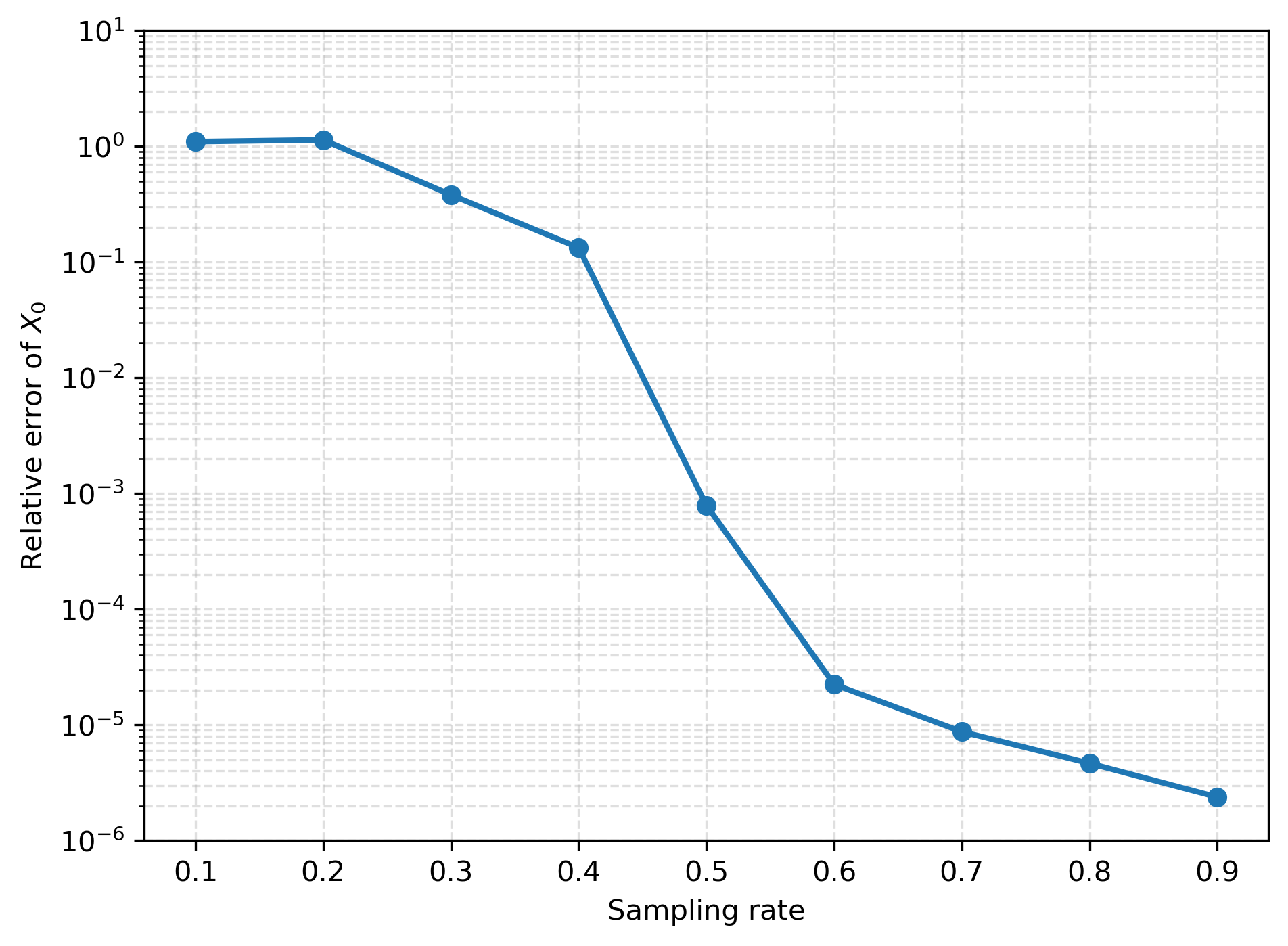}
\end{minipage}
\caption{Relative recovery errors of $A$ (left) and $X_0$ (right) versus spatial sampling rate with $n=10$, $m=20$, $r=5$, and $T=7$.}
\label{fig:synthetic-sampling-rate}
\end{figure}

We next fix the sampling rate at $m_s/n=0.5$, so that each trajectory
is observed at $m_s=5$ coordinates, and vary the observation horizon
from $T=2$ to $T=11$. The sampling locations are held fixed throughout the experiment. Thus, the observation horizon is the only experimental factor being
varied, allowing us to assess the effect of additional temporal observations.

Figure~\ref{fig:synthetic-horizon} shows that the recovery errors of
both $A$ and $X_0$ decrease overall as the observation horizon increases. Although the necessary dimension condition $rT\geq n$
already holds at the shortest horizon, it does not by itself guarantee accurate numerical recovery. The improvement is more pronounced for
shorter horizons and becomes more gradual as additional observation times are included.
\begin{figure}[t]
\centering
\begin{minipage}{0.48\textwidth}
  \centering
  \includegraphics[width=\linewidth]{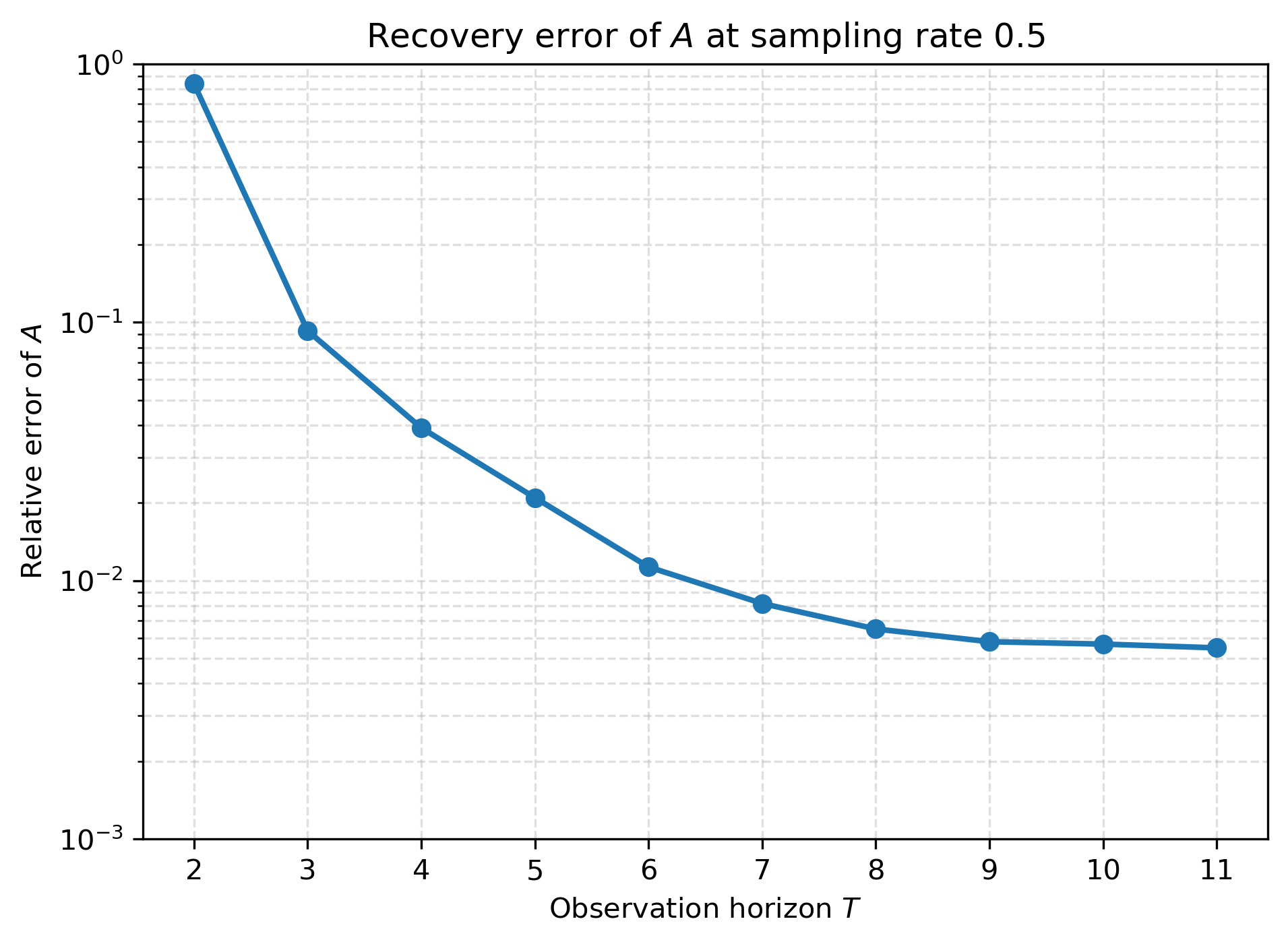}
\end{minipage}\hfill
\begin{minipage}{0.48\textwidth}
  \centering
  \includegraphics[width=\linewidth]{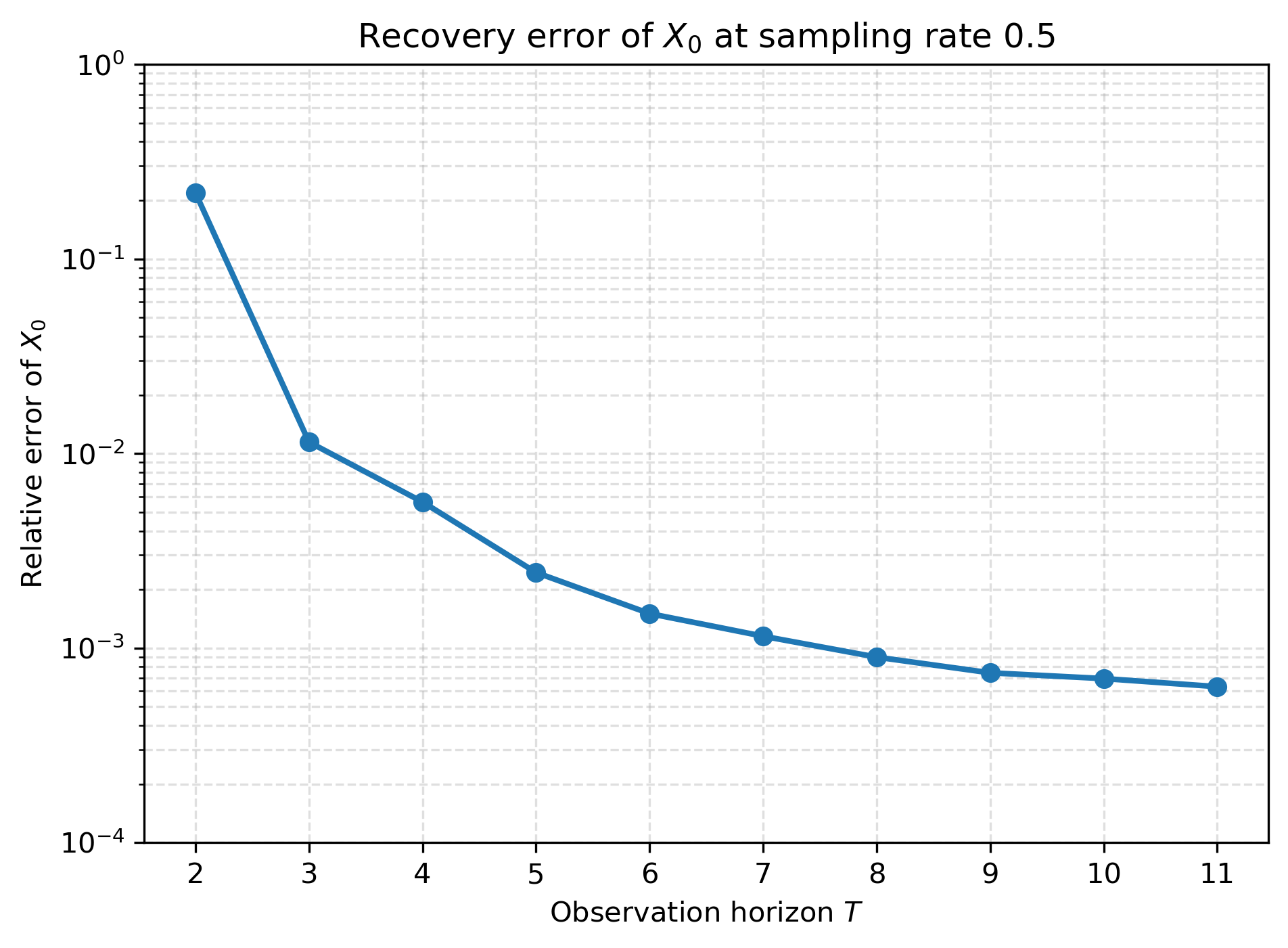}
\end{minipage}
\caption{Relative recovery errors of $A$ (left) and $X_0$ (right)
versus observation horizon $T$ with
$n=10$, $m=20$, $r=5$, and $m_s/n=0.5$.}
\label{fig:synthetic-horizon}
\end{figure}

\subsection{Sea-surface-temperature data}
\label{subsec:temperature-experiment}

We finally evaluate the proposed method on a sea-surface-temperature
data set provided by the NOAA Physical Sciences Laboratory (PSL). The data contain monthly temperature measurements at $100$ spatial
locations over $1733$ months. We select
$n=50$ locations and retain the $100$ consecutive observations from month 11 through month 110.

The selected time interval is partitioned into $m=10$ nonoverlapping
trajectory segments, each containing $10$ consecutive monthly states.
Thus, each segment has observation horizon $T=9$, and the corresponding
initial states occur at months
\[
11,21,31,\ldots,101.
\]
We use the factorization rank $r=6$. To improve numerical conditioning,
the selected temperatures are normalized by the common scale factor
$c=30.46$, their maximum absolute value. This normalization leaves the
relative error unchanged, while errors reported in degrees Celsius are
rescaled by $c$.

We compare the proposed dynamically coupled estimator with independent
factorized low-rank matrix completion. In the baseline method, each
state matrix is represented separately as
\[
X_t=L_tR_t^\top,
\qquad
\operatorname{rank}(X_t)\leq r,
\]
and the factors $L_t$ and $R_t$ are estimated using only the observed
entries of $X_t$. The baseline uses the same prescribed rank and is
also optimized with Adam. However, it does not couple the state
matrices through the relation $X_{t+1}=AX_t$ and does not estimate a
common system matrix $A$.

We consider the sampling rates
\[
\frac{m_s}{n}=0.1,0.2,\ldots,0.9.
\]
At each rate, every trajectory is observed at $m_s$ of the $50$
locations, and its sampling set remains fixed over its ten time
instances. The proposed method and the matrix-completion baseline are
given exactly the same observed entries. The sampling sets used at higher rates are not required to contain those used at lower rates. Let $\Omegaall_{\mathrm{miss}}$ denote the collection of unobserved entries over all trajectories and observation times. We evaluate the reconstruction
accuracy of the two methods on $\Omegaall_{\mathrm{miss}}$ using the mean absolute error
(MAE) and the relative error (RE), defined by
\begin{equation}
\label{eq:temperature-errors}
\operatorname{MAE}
=
\frac{1}{|\Omegaall_{\mathrm{miss}}|}
\sum_{j\in\Omegaall_{\mathrm{miss}}}
|\widehat X_j-X_j|,
\qquad
\operatorname{RE}
=
\left(
\frac{
\sum_{j\in\Omegaall_{\mathrm{miss}}}|\widehat X_j-X_j|^2
}{
\sum_{j\in\Omegaall_{\mathrm{miss}}}|X_j|^2
}
\right)^{1/2}.
\end{equation}

The results are shown in Figure~\ref{fig:temperature-reconstruction}.  At the
two lowest sampling rates, both methods have large errors.  The advantage of
the dynamical model is most pronounced at the intermediate rates
$0.3$--$0.6$.  For example, at rate $0.3$, the proposed method gives an MAE
of $1.001\,{}^\circ\mathrm C$ and an RE of $6.94\%$, compared with
$6.799\,{}^\circ\mathrm C$ and $48.24\%$ for independent matrix completion.
At rate $0.5$, the corresponding values are
$0.624\,{}^\circ\mathrm C$ and $3.86\%$ for the proposed method, versus
$1.604\,{}^\circ\mathrm C$ and $11.97\%$ for matrix completion.

At the higher rates $0.7$--$0.9$, the two methods become comparable, and
independent matrix completion is slightly more accurate.  At rate $0.9$, its
MAE and RE are $0.380\,{}^\circ\mathrm C$ and $2.37\%$, while those of the
proposed method are $0.523\,{}^\circ\mathrm C$ and $3.04\%$.  This behavior
is reasonable: when most entries are directly observed, independent
low-rank completion has sufficient spatial information and does not incur
the modeling error associated with imposing one common linear evolution on
real temperature data.  Overall, the experiment indicates that dynamical
coupling is particularly useful in the intermediate-sampling regime, where
purely spatial low-rank completion remains underdetermined but the temporal
measurements contain substantial additional information.

\begin{figure}[t]
\centering
\begin{minipage}{0.48\textwidth}
  \centering
  \includegraphics[width=\linewidth]{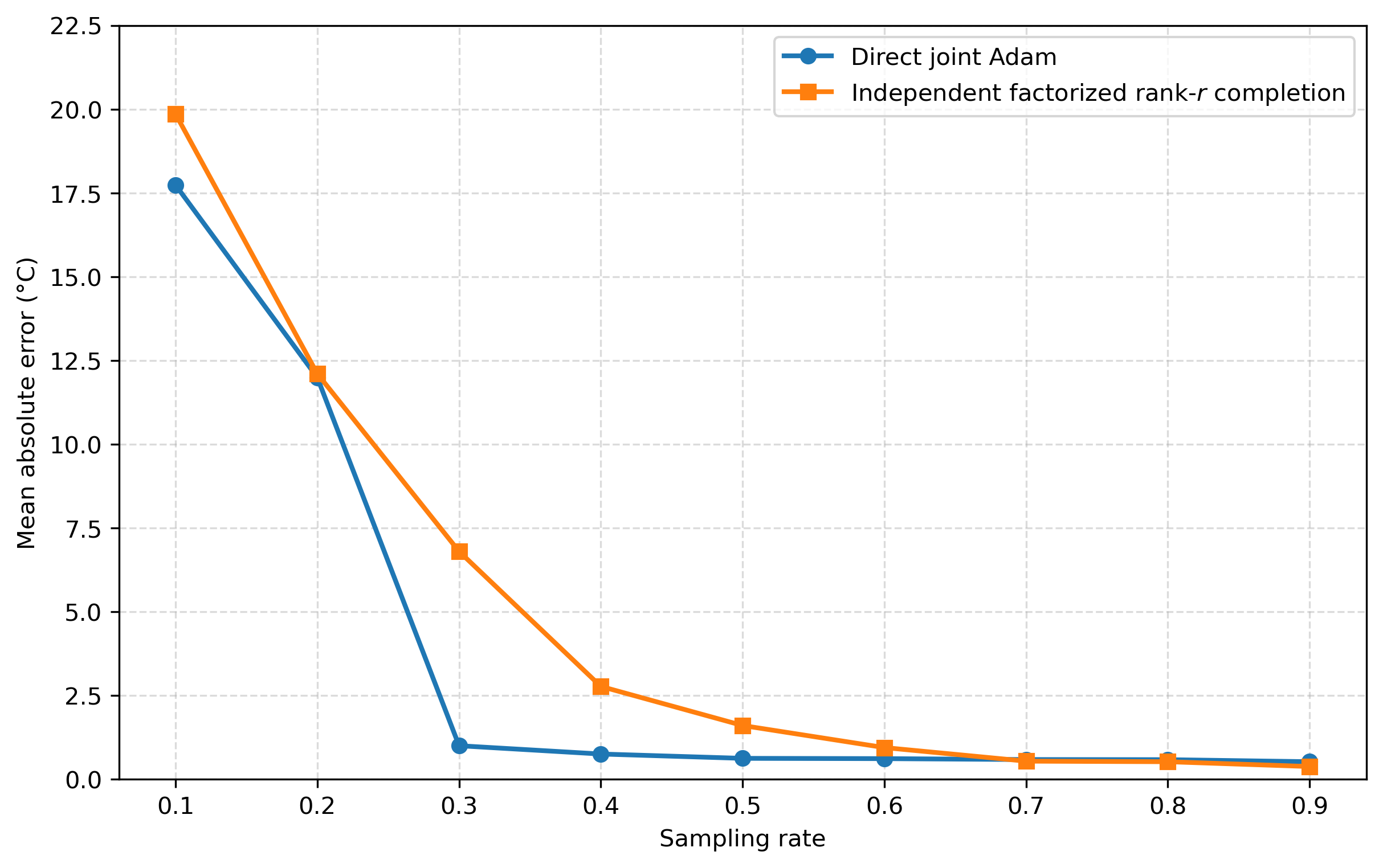}
\end{minipage}\hfill
\begin{minipage}{0.48\textwidth}
  \centering
  \includegraphics[width=\linewidth]{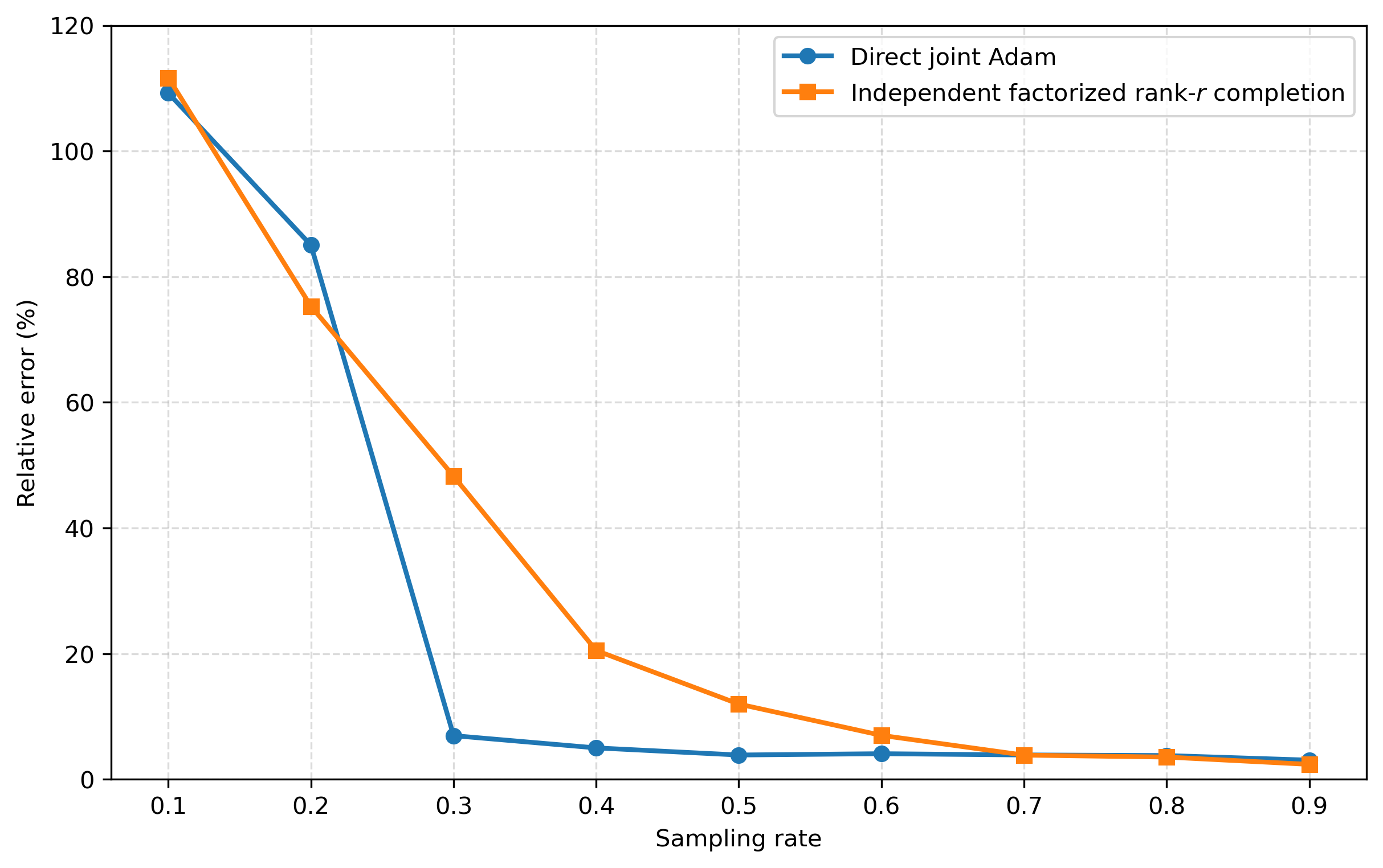}
\end{minipage}
\caption{Reconstruction of unobserved sea-surface temperatures.  Mean
absolute error (left) and relative error (right) versus spatial sampling rate.}
\label{fig:temperature-reconstruction}
\end{figure}

\appendix
\section{Proofs of the necessary conditions}
\label{app:obstructions}
\begin{proof}[Proof of Proposition~\ref{prop:nonuniqueness-fixed-S}]
For any $R\in G_{\mathcal S}$, since $S_iR=S_i$ for every $i$, it follows that
\[
S_i(RAR^{-1})^tRx_0^{(i)}
=
S_iRA^tx_0^{(i)}
=
S_iA^tx_0^{(i)},
\qquad t\geq0.
\]
Hence the pairs
$(RAR^{-1},RX_0)$ and $(A,X_0)$
produce exactly the same observations.

Now suppose that $\ker\mathcal S\neq\{0\}$. Choose
$0\neq u\in\ker\mathcal S$ and $0\neq a\in\mathbb R^n$, and define
\[
R_\varepsilon=I_n+\varepsilon ua^\top.
\]
Furthermore, we have
\[
\det R_\varepsilon
=
1+\varepsilon a^\top u.
\]
Therefore, $R_\varepsilon$ is invertible and different from $I_n$ for
every sufficiently small nonzero $\varepsilon$. Moreover, since
$u\in\ker\mathcal S$, we have $S_i u=0$ for every $i$, and hence
\[
S_iR_\varepsilon
=
S_i+\varepsilon S_iua^\top
=
S_i.
\]
Thus, $R_\varepsilon\in G_{\mathcal S}$ and can be chosen arbitrarily
close to $I_n$.

For row-selection operators, \(\ker \mathcal S=\{0\}\) is equivalent to every coordinate being sampled. If coverage fails, choose $u$ as above. Since
$\rank X_0=r\geq1$, we can choose $a$ with $a^\top X_0\neq0$. Then
\[
R_\eps X_0-X_0=\eps u(a^\top X_0)\neq0.
\]
The transformed pair in
\eqref{eq:similarity-ambiguity} is therefore distinct, has rank-$r$
initial states and identical observations, and converges to $(A,X_0)$
as $\eps\to0$.
\end{proof}

\begin{proof}[Proof of Proposition~\ref{prop:krylov-obstruction}]
Define
\[
\mathcal E
:=
\{E\in\mathbb R^{n\times n}:EK_-=0\}.
\]
Since $\rank K_-<n$, we have $\mathcal E\neq\{0\}$.

Fix $E\in\mathcal E$ and set $\widetilde A=A+E$.  We prove by induction that
\[
  \widetilde A^tX_0=A^tX_0,
  \qquad t=0,\ldots,T.
\]
The assertion is immediate at $t=0$.  If it holds at some $t\leq T-1$, then
$EA^tX_0=0.$ Therefore
\[
  \widetilde A^{t+1}X_0
  =(A+E)\widetilde A^tX_0
  =(A+E)A^tX_0
  =A^{t+1}X_0.
\]
Every nonzero $E\in\mathcal E$ gives a distinct transition matrix with the same complete trajectories.
\end{proof}

\begin{proof}[Proof of Theorem~\ref{thm:necessary-laws}]
Since $D\Psi(A,X_0)$ is injective, the dimension of its domain cannot
exceed that of its codomain. Hence
\[
(T+1)\sobs\geq n^2+r(n+m-r).
\]

Next, restrict $D\Psi(A,X_0)$ to perturbations of the form
$(\Delta A,0)$. This restricted map must also be injective. At time
$t=0$, the measurements do not depend on $A$, so their derivative with
respect to $A$ is zero. Hence, information about $A$ is provided only
by the measurements at times $t=1,\ldots,T$. These measurements contain
$T\sobs$ scalar values, whereas $\Delta A$ has $n^2$ degrees of freedom.
Injectivity therefore requires
\[
T\sobs\geq n^2.
\]

Write $X_0=B_0C^\top$ for some $C\in\mathbb R^{m\times r}$. Suppose that $\rank\mathcal O_{i,T}(A,B_0)<r$ for some $i\in[m]$. Then there
exists $0\neq h\in\mathbb R^r$ such that
\[
\mathcal O_{i,T}(A,B_0)h=0.
\]
Set
\[
\Delta X=B_0he_i^\top.
\]
This changes only the $i$-th initial state. Since
$\mathcal O_{i,T}(A,B_0)h=0$, the resulting change in its observations
is zero at every time. Thus
\[
D\Psi(A,X_0)[0,\Delta X]=0.
\]
Moreover, $\Delta X\neq0$ because $B_0$ has full column rank and
$h\neq0$. This contradicts the injectivity of $D\Psi(A,X_0)$.
Therefore,
\[
\rank\mathcal O_{i,T}(A,B_0)=r
\qquad\text{for every }i\in[m].
\]
Since $\mathcal O_{i,T}(A,B_0)$ has $(T+1)m_{s,i}$ rows, this requires $(T+1)m_{s,i}\geq r$ for every $i\in[m].$
\end{proof}

\section{Proofs of the Jacobian and local results}
\label{app:jacobian}
\begin{proof}[Proof of Proposition~\ref{prop:intrinsic-vs-gauge}]
Formula \eqref{eq:Dpi} follows from the product rule, and surjectivity onto the
tangent space follows from \eqref{eq:tangent-rank-manifold}.  It remains to
identify the kernel.  Suppose
\[
  \Delta UV^\top+U\Delta V^\top=0.
\]
Let $R_V=V(V^\top V)^{-1}$, so $V^\top R_V=I_r$.  Multiplying on the
right by $R_V$ gives
\[
  \Delta U+U\Delta V^\top R_V=0.
\]
Thus $\Delta U=UH$ with $H=-\Delta V^\top R_V$.  Substitution into the
original equality yields
\[
  U(HV^\top+\Delta V^\top)=0.
\]
Since $U$ has full column rank,
$\Delta V^\top=-HV^\top$, or $\Delta V=-VH^\top$.  The converse is
immediate.

For the equivalence, use $D\Phi=D\Psi\circ D\pi$.
If (i) holds and $D\Phi[\Delta A,\Delta U,\Delta V]=0$, then
\[
  D\Psi\paren{D\pi[\Delta A,\Delta U,\Delta V]}=0.
\]
Injectivity gives $D\pi[\Delta A,\Delta U,\Delta V]=0$, and
\eqref{eq:ker-Dpi} gives the gauge form.

Conversely, suppose (ii) holds and
$D\Psi(A,X_0)[\Delta A,\Delta X]=0$ with
$\Delta X\in T_{X_0}\mathcal M_r$.  By surjectivity of $D\pi$, choose
$\Delta U,\Delta V$ with
$\Delta X=\Delta UV^\top+U\Delta V^\top$.  Then
$(\Delta A,\Delta U,\Delta V)\in\ker D\Phi$, so (ii) implies
$\Delta A=0$ and $\Delta X=0$.
\end{proof}

\begin{proof}[Proof of Theorem~\ref{thm:generic-jacobian-rank}]
By Proposition~\ref{prop:intrinsic-vs-gauge},
\[
\mathcal G(\theta)
:=
\ker D\pi(\theta)
=
\left\{
(0,UH,-VH^\top):
H\in\mathbb R^{r\times r}
\right\}
\subseteq
\ker J(\theta).
\]
Since $\mathcal G(\theta)$ has dimension $r^2$, 
\[
\rank J(\theta)
\leq
n^2+nr+mr-r^2
=
d.
\]

Since $\Phi$ is a polynomial map, the entries of $J(\theta)$ are
polynomials in $\theta=(A,U,V)$. By the rank-certificate assumption,
\[
\rank J(\theta_0)=d.
\]
Hence, some $d\times d$ minor of $J(\theta_0)$ has a nonzero
determinant. Denote the corresponding determinant by $p(\theta)$.
Then $p$ is a polynomial and
\[
p(\theta_0)\neq0,
\]
so $p$ is not identically zero.

The zero set of a nonzero polynomial has Lebesgue measure zero.
Therefore, for almost every parameter point $\theta$, we have
$p(\theta)\neq0$. At each such point, the corresponding minor is
nonsingular, and hence
\[
\rank J(\theta)\geq d.
\]
Together with the upper bound $\rank J(\theta)\leq d$, this gives
\[
\rank J(\theta)=d,
\qquad
\dim\ker J(\theta)=r^2.
\]
Since $\mathcal G(\theta)\subseteq\ker J(\theta)$ and both spaces have
dimension $r^2$, they must be equal.
\end{proof}

\begin{proof}[Proof of Corollary~\ref{cor:local_identifiability}]
Since $\rank D\Phi(\theta_\star)=d$, some $d\times d$ minor of
$D\Phi$ is nonzero at $\theta_\star$. By continuity, this minor remains nonzero in a neighborhood $\mathcal W_0$ of $\theta_\star$. By shrinking $\mathcal W_0$ if necessary, we ensure that $U$ and $V$ remain full column rank throughout $\mathcal W_0$. For every $\theta\in\mathcal W_0$, the gauge kernel has dimension
$r^2$, so $\rank D\Phi(\theta)\leq d$. The nonzero minor gives the reverse inequality.
Hence,
\[
\rank D\Phi=d
\qquad\text{on }\mathcal W_0.
\]

By the constant-rank theorem, the local fiber
\[
\mathcal F_\star
:=
\Phi^{-1}(\Phi(\theta_\star))\cap\mathcal W_0
\]
is an $r^2$-dimensional manifold. It contains the
$r^2$-dimensional gauge orbit
\[
\mathcal O_\star
:=
\left\{
(A_\star,U_\star Q,V_\star Q^{-\top}):
Q\in\mathrm{GL}(r)
\right\}
\cap\mathcal W_0.
\]
Moreover,
\[
T_{\theta_\star}\mathcal O_\star
=
\ker D\Phi(\theta_\star)
=
T_{\theta_\star}\mathcal F_\star.
\]
The inverse function theorem therefore implies that
$\mathcal O_\star$ and $\mathcal F_\star$ coincide locally near
$\theta_\star$.

Consequently, every nearby point with the same measurements has the
form
\[
\widetilde A=A_\star,
\qquad
(\widetilde U,\widetilde V)
=
(U_\star Q,V_\star Q^{-\top})
\]
for a unique $Q\in\mathrm{GL}(r)$ close to $I_r$. The uniqueness of
$Q$ follows from the full column rank of $U_\star$.
\end{proof}

\section{Proofs for fully observed basis trajectories}
\label{app:basiss}
\begin{proof}[Proof of Theorem~\ref{thm:basis-iff}]
We first prove sufficiency and derive the reconstruction formulas. Let
$(\widetilde A,\widetilde X_0)$ be any rank-$r$ pair producing the same
observations as $(A,X_0)$. Since the trajectories indexed by
$\mathcal I$ are fully observed,
\[
\widetilde A^t\widetilde X_0[:,\mathcal I]
=
A^tX_0[:,\mathcal I]
=
B_t,
\qquad t=0,\ldots,T.
\]
Hence,
\[
K_{B,+}
=
AK_{B,-}
=
\widetilde A K_{B,-}.
\]
Since $K_{B,-}$ has full row rank,
\[
\widetilde A
=
A
=
K_{B,+}K_{B,-}^{\dagger}.
\]

At $t=0$, the fully observed basis columns give
\[
\widetilde X_0[:,\mathcal I]=B_0.
\]
Since $\rank B_0=\rank\widetilde X_0=r$, every column of
$\widetilde X_0$ belongs to $\range(B_0)$. Thus, for each
$i\notin\mathcal I$, there exists a unique
$\widetilde c_i\in\mathbb R^r$ such that
\[
\widetilde x_0^{(i)}=B_0\widetilde c_i.
\]
The observations of trajectory $i$ satisfy
\[
\mathcal O_{i,T}(A,B_0)\widetilde c_i
=
\begin{bmatrix}
y_0^{(i)}\\
\vdots\\
y_T^{(i)}
\end{bmatrix}
=
\mathcal O_{i,T}(A,B_0)c_i.
\]
Since $\mathcal O_{i,T}(A,B_0)$ has full column rank,
\[
\widetilde c_i=c_i
=
\mathcal O_{i,T}(A,B_0)^\dagger
\begin{bmatrix}
y_0^{(i)}\\
\vdots\\
y_T^{(i)}
\end{bmatrix}.
\]
Therefore, $\widetilde X_0=X_0$. Together with
$\widetilde A=A$, this proves global uniqueness.

We now prove necessity. Suppose first that
\[
s:=\rank K_{B,-}<n.
\]
Then there exists a nonzero matrix $E$ such that $EK_{B,-}=0.$ As in the proof of Proposition~\ref{prop:krylov-obstruction}, this
implies
\[
(A+E)^tB_0=A^tB_0,
\qquad t=0,\ldots,T.
\]
Since every initial state has the form $x_0^{(i)}=B_0c_i$, it follows
that
\[
(A+E)^tx_0^{(i)}=A^tx_0^{(i)}
\]
for every $i$ and $t=0,\ldots,T$. Thus, $A+E$ and $A$ produce the same
complete trajectories. 

Next, suppose that
\[
s_i:=\rank\mathcal O_{i,T}(A,B_0)<r
\]
for some $i\notin\mathcal I$. Choose
$0\neq h\in\ker\mathcal O_{i,T}(A,B_0)$ and replace $c_i$ by $c_i+h.$ Then the new initial state is $B_0(c_i+h)$, and its observations remain
unchanged because
\[
\mathcal O_{i,T}(A,B_0)h=0.
\]
All other initial states are left unchanged. The resulting initial-state
matrix still has rank $r$, since its basis columns remain equal to
$B_0$. Hence both rank conditions are necessary.

Under the full-rank conditions, the reconstruction formulas define a
smooth local map $\mathcal R$ satisfying
\[
\mathcal R\circ\Psi=\operatorname{id}.
\]
Differentiating this identity shows that $D\Psi(A,X_0)$ is injective.
Proposition~\ref{prop:intrinsic-vs-gauge} therefore gives
\[
\rank J(A,U,V)=n^2+r(n+m-r).
\]
\end{proof}

\begin{lemma}
\label{lem:generic-Krylov}
Suppose that $rT\geq n$. Then, for generic
$(A,B_0)\in\mathbb R^{n\times n}\times\mathbb R^{n\times r}$ with
$\rank B_0=r$,
\[
\rank
\begin{bmatrix}
B_0&AB_0&\cdots&A^{T-1}B_0
\end{bmatrix}
=n.
\]
\end{lemma}

\begin{proof}
It is enough to construct one pair for which the Krylov matrix has full
row rank. Set $B_0=[e_1,\ldots,e_r]$
and define
\[
Ae_j=
\begin{cases}
e_{j+r},&j+r\leq n,\\
0,&j+r>n.
\end{cases}
\]
Then
\[
A^te_j=e_{j+tr}
\]
whenever $j+tr\leq n$. Since $rT\geq n$, the columns of the Krylov
matrix include $e_1,\ldots,e_n$. Hence one of its $n\times n$ minors
is nonzero. Since this minor is a polynomial in the entries of
$(A,B_0)$, it is generically nonzero.
\end{proof}

\begin{lemma}
\label{lem:generic-observability}
Let $S\in\mathbb R^{q\times n}$ be a row-selection matrix. If $(T+1)q\geq r,$ then for generic $(A,B_0)$ with $\rank B_0=r$,
\[
\rank
\begin{bmatrix}
SB_0\\
SAB_0\\
\vdots\\
SA^TB_0
\end{bmatrix}
=r.
\]
\end{lemma}

\begin{proof}
Let $\omega_1,\ldots,\omega_q$ denote the sampled coordinates, so that
\[
Se_{\omega_a}=e_a,
\qquad a=1,\ldots,q.
\]
Set $c=\min\{q,r\}$ and choose integers
\[
1\leq \ell_a\leq T+1,
\qquad
\sum_{a=1}^c\ell_a=r.
\]
Such a choice is possible because $(T+1)q\geq r$. For each $a=1,\ldots,c$, choose standard basis vectors
\[
z_{a,0},\ldots,z_{a,\ell_a-1},
\qquad
z_{a,0}=e_{\omega_a},
\]
such that all the vectors $z_{a,s}$ are mutually distinct and
\[
Sz_{a,s}=0,
\qquad s\geq1.
\]

Define
\[
Az_{a,0}=0,
\qquad
Az_{a,s}=z_{a,s-1}\quad(s\geq1),
\]
and set $A$ to zero on the remaining standard basis vectors. Let the
columns of $B_0$ be the $r$ vectors $z_{a,s}$. Since these are distinct
standard basis vectors, $\rank B_0=r.$

For each $z_{a,s}$, the state reaches the sampled vector
$e_{\omega_a}$ exactly at time $s$. Hence, for $t=0,\ldots,T$,
\[
SA^tz_{a,s}
=
\begin{cases}
e_a,&t=s,\\
0,&t\neq s.
\end{cases}
\]
Therefore, the rows indexed by the pairs $(a,s)$ and the columns
indexed by $z_{a,s}$ form an $r\times r$ identity submatrix of the
stacked observability matrix. Its determinant is a nonzero polynomial
in the entries of $(A,B_0)$ and is therefore generically nonzero.
\end{proof}

\begin{proof}[Proof of Corollary~\ref{cor:generic-basis}]
For the fixed index set $\mathcal I$, the condition $\rank B_0=r$ is
generic. By Lemma~\ref{lem:generic-Krylov}, a maximal minor of the
Krylov matrix is a nonzero polynomial. By
Lemma~\ref{lem:generic-observability}, for each
$i\notin\mathcal I$, some $r\times r$ minor of
$\mathcal O_{i,T}(A,B_0)$ is also a nonzero polynomial.

The product of these finitely many polynomials is not identically zero.
Hence, outside its zero set, all the required rank conditions hold
simultaneously. The conclusions now follow from
Theorem~\ref{thm:basis-iff}.
\end{proof}

\section*{Declarations}

\subsection*{Conflict of Interest}
The author declares that there is no conflict of interest.

\subsection*{Funding}
The author received no financial support for this article.

\subsection*{Author Contributions}
Not applicable. This is a single-author manuscript.

\subsection*{Acknowledgements}
Not applicable.

\printbibliography
\end{document}